\documentclass[%
 reprint,
 amsmath,amssymb,
 aps,
 nofootinbib,
 superscriptaddress
]{revtex4-1}

\usepackage[dvipdfmx]{graphicx}% Include figure files
\usepackage[dvipdfmx]{xcolor}
\usepackage{dcolumn}% Align table columns on decimal point
\usepackage{bm}% bold math
\usepackage{braket}
\usepackage{amsmath}
\usepackage{amssymb}
\usepackage{amsthm}
\usepackage{enumerate}
\usepackage{booktabs}
\usepackage{txfonts}
\usepackage[update,prepend]{epstopdf}
\usepackage[english]{babel}
\usepackage{listings}
\usepackage{qcircuit}
\usepackage[subrefformat=parens]{subcaption}
\usepackage{caption}
\usepackage{multirow}
\newcommand{\argmin}{\mathop{\rm argmin}\limits}

\newtheorem{assumption}{Assumption}
\newtheorem{theorem}{Theorem}
\newtheorem{lemma}{Lemma}

\allowdisplaybreaks[1]

\makeatletter
\def\hlinewd#1{%
	\noalign{\ifnum0=`}\fi\hrule \@height #1 %
	\futurelet\reserved@a\@xhline}
\makeatother

\begin{document}

\preprint{APS/123-QED}

\title{Linear Regression by Quantum Amplitude Estimation and its Extension to Convex Optimization}

\author{Kazuya Kaneko}
\affiliation{Mizuho-DL Financial Technology Co., Ltd.\\ 2-4-1 Kojimachi, Chiyoda-ku, Tokyo, 102-0083, Japan}

\author{Koichi Miyamoto}
\email{koichi.miyamoto@qiqb.osaka-u.ac.jp}
\affiliation{Center for Quantum Information and Quantum Biology, Osaka University \\ 1-3 Machikaneyama, Toyonaka, Osaka, 560-8531, Japan}
\affiliation{Mizuho-DL Financial Technology Co., Ltd.\\ 2-4-1 Kojimachi, Chiyoda-ku, Tokyo, 102-0083, Japan}

\author{Naoyuki Takeda}
\affiliation{Mizuho-DL Financial Technology Co., Ltd.\\ 2-4-1 Kojimachi, Chiyoda-ku, Tokyo, 102-0083, Japan}

\author{Kazuyoshi Yoshino}
\affiliation{Mizuho-DL Financial Technology Co., Ltd.\\ 2-4-1 Kojimachi, Chiyoda-ku, Tokyo, 102-0083, Japan}

\date{\today}

\begin{abstract}
Linear regression is a basic and widely-used methodology in data analysis.
It is known that some quantum algorithms efficiently perform least squares linear regression of an exponentially large data set.
However, if we obtain values of the regression coefficients as classical data, the complexity of the existing quantum algorithms can be larger than the classical method.
This is because it depends strongly on the tolerance error $\epsilon$: the best one among the existing proposals is $O(\epsilon^{-2})$.
In this paper, we propose a new quantum algorithm for linear regression, which has a complexity of $O(\epsilon^{-1})$ and keeps a logarithmic dependence on the number of data points $N_D$.
In this method, we overcome bottleneck parts in the calculation, which take the form of the sum over data points and therefore have a complexity proportional to $N_D$, using quantum amplitude estimation, and other parts classically.
%We also discuss the relationship between our method to the previously proposed method, which is based on gradient descent and uses QAE for gradient estimation similarly to our method.
Additionally, we generalize our method to some class of convex optimization problems.
\end{abstract}

\pacs{Valid PACS appear here}% PACS, the Physics and Astronomy
                              
\maketitle

\section{\label{sec:intro}Introduction}

Following the rapid advance of quantum computing technology, many quantum algorithms have been proposed and their applications to the wide range of practical problems have been studied in recent research.
One prominent example is linear regression.
Linear regression, which is based on the least squares method in many cases, is a basic and ubiquitous methodology for many fields in natural and social sciences.
There are some quantum algorithms for linear regression~\cite{Wiebe,Wang,Kerenidis,Yu,Schuld,Chakraborty} with complexity depending on the number of data points $N_D$ as $O({\rm polylog}(N_D))$ \footnote{There are also quantum-inspired classical methods~\cite{Chia,Gilyen,Shao}.}.
This denotes the exponential speedup compared with the naive classical method explained in Sec. \ref{sec:LR}, whose complexity is proportional to $N_D$.
%Therefore, when $N_D$ is large, quantum methods are advantageous over classical ones.

However, despite this, the existing quantum methods are not necessarily more beneficial than classical ones, when we want to obtain the values of the optimized regression coefficients {\it as classical data}, and $N_D$ is {\it mid-sized}, say $O(10^3-10^5)$.
Roughly speaking, in the existing methods such as the first proposed one~\cite{Wiebe}, which is based on the Harrow-Hassidim-Lloyd (HHL) algorithm~\cite{HHL} for solving systems of linear equations, the authors create quantum states in which the values of the coefficients are encoded and read out in classical form.
Therefore, it is inevitable that the estimated coefficients are accompanied by errors, and high-accuracy estimation leads to large complexity.
As far as we know, the existing method with the best complexity with respect to the tolerance error $\epsilon$ is that in \cite{Wang}: the complexity of estimating coefficients with additive error at most $\epsilon$ is
\begin{equation}
O\left(\frac{d^{5/2}\kappa^3}{\epsilon^2}{\rm polylog}\left(\frac{d\kappa}{\epsilon}\right)\right), \label{eq:CompWang}
\end{equation}
where $d$ is the number of the coefficients (or the explanatory variables) and $\kappa$ is the condition number of the design matrix (defined in Sec. \ref{sec:LR}).
%On the other hand, the classical method has complexity $O\left(d^2N_D\right)$, as discussed in Sec. \ref{sec:LR}.
On the other hand, a naive classical method has complexity $O\left(d^2N_D\right)$.
Therefore, assuming the prefactors in these expressions of complexities are comparable, under an ordinary situation $\epsilon \sim 10^{-3},d \sim 10, \kappa \sim 1$, the minimum $N_D$ for which the quantum method is advantageous over classical one is $N_D\sim 10^6$.
%The number of available data points are often of the order of $10^3$ or $10^4$ and in such cases the exiting quantum methods are inferior to classical ones.
Although in some problems $N_D$ is of such an order or much larger, cases where $N_D \sim 10^3$ or $10^4$ are also ubiquitous, and in such cases the exiting quantum methods are inferior to classical ones.

The mid-sized regression problems are often time-consuming and need to be sped up, although we might naively think that such problems can be solved by classical computers in a short time.
For example, if such regressions are repeated in the whole of the calculation flow, the total computational time can be long.
%It might be considered that regression problems with a small number of data points can be solved by classical computers in short time and therefore is out of target of quantum methods from the beginning.
%However, a small-sized regression is sometimes time-consuming because it is repeated in the whole of the calculation flow. 
One such case is the {\it least square Monte Carlo (LSM)}~\cite{Longstaff}.
LSM is a methodology used in pricing financial derivatives\footnote{Financial derivatives, or simply derivatives, are financial contracts in which two parties exchanges cash flows whose amounts are determined by the price of some assets. For textbooks on derivatives and their pricing, we refer to \cite{Hull,Shreve}.} with an early exercise option, that is, the contract term stating that either of two parties in the contract can terminate it before the final maturity.\footnote{For the details of LSM, see \cite{Longstaff}.}
In Monte Carlo pricing, we generate many (say, $O(10^4-10^5)$) paths of the random time evolution of underlying assets, and we estimate the price as the expected cash flow.
In the case of early-exercisable products, we have to determine the optimal exercise time for each path. %, and LSM is the linear-regression-based method for this.
In LSM, we approximate the continuation value of the derivative at each exercise date by linear regression using certain functions of underlying asset prices as explanatory variables, whose number is typically $d \sim 10$.
%Regression must be done for every exercise date in one contract and so the repetition number becomes large for long-dated contracts.
%Furthermore, large banks have numerous such contracts and must price them in the daily business.
Since regression is done many times in pricing one contract and banks have numerous contracts, they have huge complexity in total and are meaningful targets of quantum speedup, even though each of them is mid-sized.

Based on the above motivation, in this paper, we present a new quantum algorithm for linear regression, focusing on reducing the order of the inverse of $\epsilon$ in the expression of the complexity.
In our method, unlike existing methods, we do not perform all of the calculations on a quantum computer.
Instead, we use a quantum computer {\it only to perform a bottleneck part in the naive classical method}.
That is, we perform some intermediate calculation, which is the sum over $N_D$ terms and therefore has complexity proportional to $N_D$ naively, by quantum algorithm.
Then, we classically solve the $d$-dimensional system of linear equations, whose coefficients and constant terms are outputs of the preceding quantum computation, and we obtain the regression coefficients.
A classical computer can perform this step in negligible time for $d\sim 10$, without inducing any error.
Also note that we naturally obtain the regression coefficients as classical data.

%These bottleneck calculations have the form of the sum over the data points.
To speed up the bottleneck sums, we use {\it quantum amplitude estimation (QAE)}~\cite{Brassard,Suzuki,Aaronson,Grinko,Nakaji,Brown,Tanaka}, which is also used to speed up Monte Carlo integration~\cite{Montanaro,Suzuki}.
Then, as discussed in detail in Sec. \ref{sec:LRbyQAE}, we can perform the sums with complexity
\begin{equation}
	O\left(\max\left\{\frac{d^{3/2}\kappa^4}{\epsilon},d\kappa^2\right\}\times d^2\log (d)\right), \nonumber
\end{equation}
and the total complexity of the whole of our method is almost equal to this.
Compared with the naive classical method, whose complexity is $O(d^2N_D)$, our method accomplishes speedup with respect to $N_D$. %if $d^{3/2}\kappa^4/\epsilon<N_D$.
In addition, compared with the existing quantum methods, our method improves the dependency of the complexity on $\epsilon$ from $\epsilon^{-2}$ in (\ref{eq:CompWang}) to $\epsilon^{-1}$.
Unfortunately, in our method, the dependencies on $d$ and $\kappa$ are worse than those of (\ref{eq:CompWang}) and the naive classical method.
Therefore, our method is more suitable to the situation such that
\begin{equation}
	d, \kappa \ll \frac{1}{\epsilon} \ll N_D,
\end{equation}
which includes typical mid-sized problems such as LSM.

In addition to linear regression, we generalize our method to some class of convex optimization.
As we see in Sec. \ref{sec:extension}, linear regression can be considered as an optimization problem of the sum of quadratic functions solved by Newton's method.
Inspired by this, we consider a convex optimization problem where an objective function is written as sums of many terms
%with same structure as in
similarly to linear regression, and we present a quantum algorithm of Newton's method, in which calculation of the gradient and the Hessian is sped up by QAE.
We also estimate complexity to obtain a solution with a desired error level under some mathematical assumptions which are usually made in the convergence analysis of Newton's method.

The remaining parts of this paper are organized as follows.
Sec. \ref{sec:Preliminary} is a preliminary section.
In Sec. \ref{sec:Notation}, we present some notations and definitions of quantities necessary for the later discussion.
In Sec. \ref{sec:LR}, we briefly review linear regression and its classical method.
In Sec. \ref{sec:LRbyQAE}, we explain our method in detail.
%In Sec. \ref{sec:extension}, we discuss extension of our method to more general optimization problems, clarifying the difference between our method and previous papers~\cite{Shao,Liu}.
In Sec. \ref{sec:extension}, we discuss the extension of our method to some convex optimization problems.
Sec. \ref{sec:summary} summarizes the paper.

\section{\label{sec:Preliminary}Preliminary}
\subsection{\label{sec:Notation}Notations and Definitions}

In this subsection, we present some notations and definitions about linear algebra for later use.
For details, refer to textbooks on linear algebra, e.g., \cite{Golub}.

For $\vec{v}=(v_1,...,v_n)^T\in \mathbb{R}^n$, we define the Euclidean norm $\left\|\vec{v}\right\|$ and the max norm $\left\|\vec{v}\right\|_\infty$ as follows:
\begin{equation}
\left\|\vec{v}\right\| := \sqrt{\sum_{i=1}^{n}|v_i|^2}, \ 
\left\|\vec{v}\right\|_\infty := \max\{|v_1|,...,|v_n|\}.
\end{equation}
%The Euclid norm $\left\|\vec{v}\right\|_2$ is most commonly used, and so we write it as $\left\|\vec{v}\right\|$ simply.

%We also define the induced norm for $A=(a_{ij})_{\substack{1\le i \le m \\ 1\le j \le n}}\in \mathbb{R}^{m\times n}$ as follows
%\begin{equation}
%\left\|A\right\|_p := \max_{\substack{\vec{v}\in \mathbb{R}^n \\ \left\|\vec{v}\right\|_p\ne 0}} \frac{\left\|A\vec{v}\right\|_p}{\left\|\vec{v}\right\|_p}.
%\end{equation}
%Also for matrices, we hereafter write $\left\|A\right\|_2$ as $\left\|A\right\|$.
%The norms for $p=2$ and $\infty$ can be represented as follows
%\begin{eqnarray}
%\left\|A\right\|_1 & = & \max_{j=1,...,n} \sum_{i=1}^{m} |a_{ij}|, \nonumber \\
%\left\|A\right\| & = & \sigma_{\rm max}(A), \nonumber \\
%\left\|A\right\|_\infty & = & \max_{i=1,...,m} \sum_{j=1}^{n} |a_{ij}|,
%\end{eqnarray}
%where $\sigma_{\rm max}(A)$ is the maximum singular value of $A$.
%Here, the singular values of $A$ are the square roots of the eigenvalues of ${}^t\!AA$, which is positive-semidefinite.
%In the later sections, we mainly consider norms of symmetric matrices.
%If $A$ is an $n\times n$ symmetric matrix, the following relationship is satisfied
%\begin{equation}
%\left\|A\right\| \le \left\|A\right\|_\infty.
%\end{equation}
We also define some types of norms for a matrix $A=(a_{ij})_{\substack{1\le i \le m \\ 1\le j \le n}}\in \mathbb{R}^{m\times n}$.
The first one is the {\it spectral norm}, which is defined as %through the vector Euclidean norm as
\begin{equation}
\left\|A\right\| := \max_{\substack{\vec{v}\in \mathbb{R}^n \\ \left\|\vec{v}\right\|\ne 0}} \frac{\left\|A\vec{v}\right\|}{\left\|\vec{v}\right\|}, \label{eq:spnorm}
\end{equation}
%This norm can be represented as
%\begin{eqnarray}
%\left\|A\right\| = \sigma_{\rm max}(A),
%\end{eqnarray}
and equal to $\sigma_{\rm max}(A)$, the maximum singular value of $A$.
%Here, the singular values of $A$ are the square roots of the eigenvalues of $A^TA$, which is positive-semidefinite.
%The following relationship is immediately derived from the definition (\ref{eq:spnorm})
%\begin{eqnarray}
%\|A\vec{v}\| \le \|A\|\cdot\|\vec{v}\|.
%\end{eqnarray}
Secondly, we define the {\it Frobenius norm} as
\begin{equation}
\left\|A\right\|_F := \left(\sum_{i=1}^m \sum_{j=1}^n |a_{ij}|^2\right)^{1/2},
\end{equation}
which satisfies $\left\|A\right\| \le \left\|A\right\|_F$.
%The following relationship is well-known:
%\begin{equation}
%\left\|A\right\| \le \left\|A\right\|_F.
%\end{equation}

Furthermore, we introduce the {\it condition number}, which measures the sensitivity of a linear operator to the errors in the input.
For a full-rank matrix $A\in \mathbb{R}^{m\times n}$, the condition number is defined as
\begin{equation}
\kappa(A) := \frac{\sigma_{\rm max}}{\sigma_{\rm min}},
\end{equation}
that is, the ratio of the maximum singular value $\sigma_{\rm max}$ of $A$ to the minimum one $\sigma_{\rm min}$.
If $A$ is invertible, 
\begin{equation}
\kappa(A) = \left\|A\right\| \left\|A^{-1}\right\|.
\end{equation}

\subsection{\label{sec:LR}Linear regression}

We here define the problem of linear regression.
Assume that we have $N_D$ data points $\mathcal{D}:=\{(\vec{x}_i,y_i)\}_{i=1,..,N_D}$, each of which consists of a vector of $d$ explanatory variables\footnote{If we want to consider the intercept, we include a dummy variable (say $x^{(1)}_i$) in each explanatory variable vector and set it to 1: $x^{(1)}_1=...=x^{(1)}_{N_D}=1$.} $\vec{x}_i=(x^{(1)}_i,...,x^{(d)}_i)^T\in \mathbb{R}^d$ and an objective variable $y_i\in \mathbb{R}$.
Linear regression attempts to fit the linear combination of the explanatory variables to the objective variable, that is, finding $\vec{a}$ such that
\begin{equation}
\vec{y} \simeq X \vec{a}.
\end{equation}
Here, $a_1,...,a_d\in \mathbb{R}$ are model parameters called {\it regression coefficients} and $\vec{a}:=(a_1,...,a_d)^T\in \mathbb{R}^d$.
$\quad\vec{y}:=(y_1,...,y_{N_D})^T$ and 
\begin{equation}
X := \left(
\begin{matrix}
x^{(1)}_1 & \cdots & x^{(d)}_1 \\
\vdots & \ddots & \vdots \\
x^{(1)}_{N_D} & \cdots & x^{(d)}_{N_D}
\end{matrix}
\right)
\label{eq:designmat}
\end{equation}
is called the {\it design matrix}.
In this paper, as in many cases, we determine $\vec{a}$ by the {\it least squares method}, that is,
\begin{equation}
\vec{a} = \argmin_{\vec{a}^\prime} \left\|\vec{y}-X\vec{a}^\prime\right\|^2. \label{eq:LS}
\end{equation}
As is well known (again, see \cite{Golub} for example), the solution of (\ref{eq:LS}) is given by
\begin{equation}
\vec{a} = \left(X^TX\right)^{-1}X^T\vec{y}, \label{eq:LSSol}
\end{equation}
where we assume that $X$ is full-rank, as stated again in Sec. \ref{sec:ass} as Assumption \ref{ass:fullrank}.

Here, let us introduce some symbols.
We define $d\times d$ matrix $W$,
\begin{equation}
	W := \frac{1}{N_D}X^TX, \label{eq:WMat}
\end{equation}
whose $(i,j)$-element is
\begin{equation}
w_{ij} = \frac{1}{N_D}\sum_{k=1}^{N_D} x^{(i)}_k x^{(j)}_k. \label{eq:w}
\end{equation}
$W$ is invertible since we assumed that $X$ is full-rank.
We also define the $d$-dimensional vector $\vec{z}$,
\begin{equation}
	\vec{z}:=\frac{1}{N_D}X^T\vec{y}, \label{eq:zvec}
\end{equation}
whose $i$-th element is
\begin{equation}
z_i = \frac{1}{N_D}\sum_{k=1}^{N_D} x^{(i)}_k y_k. \label{eq:z}
\end{equation}
In (\ref{eq:WMat}) and (\ref{eq:zvec}), the prefactor $1/N_D$ is just for the later convenience.
Using these, (\ref{eq:LSSol}) becomes
\begin{equation}
	\vec{a} = W^{-1}\vec{z}. \label{eq:LSSolScale}
\end{equation}

We here call the following classical computation {\it the naive classical method}: calculate $w_{ij}$'s and $z_i$'s by simply repeating multiplications and additions $N_D$ times, literally following the definitions (\ref{eq:w}) and (\ref{eq:z}), and then solve the $d$-dimensional system of linear equations (\ref{eq:LSSolScale}) to find $\vec{a}$.
Let us discuss the complexity of this method.
Since we are considering the situation $d\ll N_D$, we focus only on the major contributions with respect to $N_D$.
Calculating one $w_{ij}$ or $z_i$ takes complexity of $O(N_D)$.
Since the numbers of $w_{ij}$'s and $z_i$'s are $O(d^2)$ and $O(d)$, respectively, and their total is $O(d^2)$, the total complexity of calculating all of $w_{ij}$'s and $z_i$'s is $O(d^2N_D)$.
Then, solving (\ref{eq:LSSolScale}) takes the negligible complexity, since it does not depend on $N_D$.
For example, even if we use the elementary method such as row reduction, the complexity is $O(d^3)$.
In summary, the complexity of the naive classical method is $O(d^2N_D)$, which dominantly comes from calculation of $w_{ij}$'s and $z_i$'s.

\section{\label{sec:LRbyQAE}Linear regression by quantum amplitude estimation}

\subsection{\label{sec:QAE}Quantum Amplitude Estimation}

In this section, we present our method for linear regression.
Before this, let us review the outline of QAE briefly.

QAE is the algorithm to estimate a probability amplitude of a marked state in a superposition.
Consider the system consisting of some qubits.
We set the system to the initial state where all qubits are set to $\ket{0}$ and write such a state as $\ket{0}$ for simplicity.
Then, we assume that there exists an unitary transformation $A$ on the system such that
\begin{equation}
A\ket{0} = a\ket{\Psi} + \sqrt{1-a^2}\ket{\Psi_\perp},
\end{equation}
where $\ket{\Psi}$ is the `marked state', $\ket{\Psi_\perp}$ is a state orthogonal to $\ket{\Psi}$ and $0<a<1$.
Typically, $\ket{\Psi}$ and $\ket{\Psi_\perp}$ are the states where a specific qubit takes $\ket{1}$ and $\ket{0}$ respectively.
In addition to $A$, we use the following unitary operators $S_0$ and $S_\Psi$, which are defined as
\begin{equation}
S_0\ket{\phi} =
\begin{cases}
-\ket{0} & ;  {\rm if} \ \ket{\phi} = \ket{0} \\
\ket{\phi} & ;  {\rm if} \ \ket{\phi} {\rm is} \ {\rm orthogonal} \ {\rm to} \ \ket{0}
\end{cases}
,
\end{equation}
\begin{equation}
S_\Psi\ket{\phi} =
\begin{cases}
-\ket{\Psi} &; {\rm if} \ \ket{\phi} = \ket{\Psi} \\
\ket{\phi} &; {\rm if} \ \ket{\phi} {\rm is} \ {\rm orthogonal} \ {\rm to} \ \ket{\Psi}
\end{cases}
.
\end{equation}
$S_0$ can be constructed using a multi-controlled Toffoli gate, and $S_\Psi$ is simply a controlled-$Z$ gate if $\ket{\Psi}$ is defined as the state where a specific qubit is $\ket{1}$.
Then, defining
\begin{equation}
Q:=-AS_0A^{-1}S_\Psi, \label{eq:Q}
\end{equation}
we can construct a quantum algorithm (see \cite{Brassard} for details) which makes
\begin{equation}
O\left(\frac{1}{\epsilon}\right) \label{eq:QAEcomp}
\end{equation}
uses of $Q$ (therefore, $O(1/\epsilon)$ uses of $A$) and outputs an estimate of $a$ with an $\epsilon$-additive error.

We make a comment here on success probability.
In the algorithm of QAE~\cite{Brassard}, the success probability, that is, the probability that the algorithm outputs the estimation with the desired additive error, is not 1 but lower-bounded by $8/\pi^2$.
However, we can enhance the success probability to an arbitrary level $1-\gamma$, where $\gamma\in (0,1)$, by repeating QAE $O\left(\log (\gamma^{-1})\right)$ times.
That is, taking the median of the results in the $O\left(\log (\gamma^{-1})\right)$ runs of QAE, we can obtain the estimation with the additive error $\epsilon$ with probability $1-\gamma$~\cite{Jerrum,Montanaro}.
Considering this point, we can write the number of calls to $A$ in repeating QAE with an $\epsilon$-additive error and a success probability larger than $1-\gamma$ as
\begin{equation}
	O\left(\frac{\log (\gamma^{-1})}{\epsilon}\right). \label{eq:QAEcompRep}
\end{equation}
If we set $1-\gamma$ to some fixed value, say 99\%, (\ref{eq:QAEcompRep}) is reduced to (\ref{eq:QAEcomp}).

\subsection{\label{sec:ass}Assumptions}

Next, we present some assumptions which are necessary for the method.
The first one is as follows.
\begin{assumption}
	The following oracles $P_x$ and $P_y$ are available:
	\begin{equation}
	P_x:\ket{i}\ket{k}\ket{0} \mapsto \ket{i}\ket{k}\ket{x^{(i)}_k}, \label{eq:Px}
	\end{equation}
	\begin{equation}
	P_y:\ket{k}\ket{0} \mapsto \ket{k}\ket{y_k}, \label{eq:Py}
	\end{equation}
	for any $i\in \{1,...,d\}$ and $k\in\{1,...,N_D\}$.
	\label{ass:oracle}
\end{assumption}
Here and hereafter, for a number $x$, the ket $\ket{x}$ corresponds to a computational basis state on a quantum register where the bit string represents the binary representation of $x$.
(\ref{eq:Px}) and (\ref{eq:Py}) mean that $P_x$ and $P_y$ output the element of $X$ and $\vec{y}$, respectively, for the specified index.
Previous papers \cite{Wiebe,Wang,Kerenidis,Yu,Schuld,Shao,Liu,Chakraborty} also assume such oracles.
We can construct $P_x$ and $P_y$ if quantum random access memories (QRAMs)~\cite{Giovannetti} are available.

The second assumption is just a reproduction.
\begin{assumption}
	$X$ defined as (\ref{eq:designmat}) is full-rank. \label{ass:fullrank}
\end{assumption}
Because of this, $\kappa(X)$, the condition number of $X$, can be defined, and $W$ defined as (\ref{eq:w}) is invertible.
Hereafter, we simply write $\kappa(X)$ as $\kappa$.

The third assumption is as follows.
\begin{assumption}
	\begin{equation}
	0 \le x^{(i)}_k \le 1, 0 \le y_k \le 1
	\end{equation}
	for any $i\in \{1,...,d\}$ and $k\in\{1,...,N_D\}$.
	\label{ass:bound}
\end{assumption}
That is, we assume that the explanatory variables and the objective variable are bounded by $0$ and $1$.
Besides, we make the fourth assumption as follows:
\begin{assumption}
	There is a positive number $c$ (say, $\frac{1}{2}$), which is independent of $N_D, \epsilon, \kappa$ and $d$, such that
	\begin{equation}
	\forall i\in\{1,...,d\}, w_{ii}=\frac{1}{N_D}\sum_{k=1}^{N_D} \left(x^{(i)}_k\right)^2 > c.
	\end{equation}
	\label{ass:scale}
\end{assumption}
Since $w_{ii}\le 1$ is immediately derived from Assumption \ref{ass:bound}, Assumption \ref{ass:scale} means that $c<w_{ii}\le1$.

Although Assumption \ref{ass:bound} and \ref{ass:scale} are too strong seemingly, they should be assumed for successful regression, in either the classical or the quantum way, for the following reasons.
First, we note that Assumption \ref{ass:bound} is satisfied if we know the bounds for $x_k^{(i)}$ and $y_k$ in advance and can rescale them.
That is, for $i\in \{1,...,d\}$, although in general $x^{(i)}_k$'s are not in $[0,1]$, we can redefine
\begin{equation}
\tilde{x}^{(i)}_k := \frac{x^{(i)}_k-L_i}{U_i - L_i}
\end{equation}
as $x^{(i)}_k$ if we know $L_i,U_i$ such that
\begin{equation}
\forall k\in\{1,...,N_D\}, L_i \le x^{(i)}_k \le U_i,
\end{equation}
Then, redefined $x^{(i)}_k$'s are in $[0,1]$.
Similarly, redefining
\begin{equation}
\tilde{y}_k:=\frac{y_k-L_y}{U_y - L_y}
\end{equation}
as $y_k$ with $L_y,U_y$ such that
\begin{equation}
\forall k\in\{1,...,N_D\}, L_y \le y_k \le U_y
\end{equation} 
leads to $0\le y_k \le 1$.
Assumption \ref{ass:bound} is then satisfied.
Besides, if we know the bounds which are not too far from the typical scale of the original $x^{(i)}_k$'s, that is, if we can take $L_i,U_i$ such that $|L_i|\sim |U_i| \sim |x^{(i)}_k|$ for most $k$'s, Assumption \ref{ass:scale} is satisfied.
In summary, Assumption \ref{ass:bound} and \ref{ass:scale} are naturally satisfied if we know the typical scales of $x^{(i)}_k$'s and $y_k$'s.
Practically, we {\it must} know the typical scales, since we have to address the problem of {\it outliers}.
Data sets often contain points whose explanatory and/or objective variables are much larger than those of others, because of various reasons (for example, misrecord).
Such points are called outliers.
It is widely known that outliers lead to inaccurate regression and so we have to address them.
Typically, we omit them from data points used for regression or replace the values of the explanatory and/or objective variables out of some range with the upper or lower bound of the range.
For such a preprocess, we have to know the typical scales of the variables.
In fact, the previous paper \cite{Wang} makes assumptions similar to Assumption \ref{ass:bound} and \ref{ass:scale}.
Mentioning the necessity of preprocessing outliers, it assumes that the design matrix and the objective variable vector do not contain the extraordinarily large elements.

\subsection{Details of our method}

We now explain our method in detail.
First, we present a lemma on the error in the solution of a system of linear equations where coefficients and constant terms contain errors.
\begin{lemma}
	%Define $W,\vec{z}$ and $\vec{a}$ as (\ref{eq:w}), (\ref{eq:z}) and (\ref{eq:LSSolScale}), respectively, and 
	Let Assumptions 2 to 4 be satisfied.
	For given symmetric $\hat{W}\in \mathbb{R}^{d\times d}$ and $\vec{\hat{z}}\in \mathbb{R}^d$, consider a system of linear equations
	\begin{equation}
	\hat{W}\vec{\hat{a}}=\vec{\hat{z}}, \label{eq:eqForhata}
	\end{equation}
	where $\vec{\hat{a}}\in \mathbb{R}^d$.
    For a given $\epsilon>0$, if each element of $\delta W:=\hat{W}-W$ and $\delta\vec{z}:=\vec{\hat{z}}-\vec{z}$ has an absolute value smaller than $\epsilon^\prime$ such that
    \begin{equation}
    	\epsilon^\prime < \min \left\{\frac{c}{d\kappa^2}, \frac{c^2\epsilon}{2d^{3/2}\kappa^4}\right\} \label{eq:epsPrimeCond}
    \end{equation}
    then %$\hat{W}$ is invertible and
    $\vec{\hat{a}}$ is uniquely determined by solving Eq. (\ref{eq:eqForhata}), i.e., $\vec{\hat{a}}=\hat{W}^{-1}\vec{\hat{z}}$, and it becomes an $O(\epsilon)$-additive approximation of $\vec{a}$, which means that
    \begin{equation}
    \|\vec{\hat{a}} - \vec{a}\|_\infty = O\left(\epsilon\right). \label{eq:erreval}
    \end{equation}
	\label{lem}
\end{lemma}
The proof is given in Appendix \ref{sec:app}.
This lemma means that, if we want a solution of a system of linear equations with an $O(\epsilon)$-additive error, it is sufficient to calculate coefficients and constant terms with an additive error $\epsilon^\prime$ satisfying (\ref{eq:epsPrimeCond}).

We therefore propose the following method for linear regression: we first estimate $W$ in (\ref{eq:w}) and $\vec{z}$ in (\ref{eq:z}) by a quantum method with $\epsilon^\prime$-additive error and then calculate (\ref{eq:LSSolScale}) by some classical method.
Since classical methods basically introduces no additional error, we can obtain a solution with $O(\epsilon)$-additive error.

Then, we state a theorem on the complexity of our method, presenting the concrete procedure in the proof.

\begin{theorem}
	Given $\epsilon>0$, accesses to oracles $P_x$ and $P_y$ which satisfy Assumptions 1, and $\{x^{(i)}_k\}_{\substack{i=1,...,d \\ k=1,...,N_D}},\{y_k\}_{k=1,...,N_D}$ which satisfy Assumption \ref{ass:fullrank} to \ref{ass:scale}, there is a quantum algorithm that makes
	\begin{equation}
	O\left(\max\left\{\frac{d^{3/2}\kappa^4}{\epsilon},d\kappa^2\right\}\times d^2\log(d)\right) \label{eq:numPx}
	\end{equation}
	uses of $P_x$ and
	\begin{equation}
	%O\left(\frac{d^{5/2}\kappa^4\log(\gamma^{-1})}{\epsilon}\right) \label{eq:numPy}
	O\left(\max\left\{\frac{d^{3/2}\kappa^4}{\epsilon},d\kappa^2\right\}\times d\log(d)\right) \label{eq:numPy}
	\end{equation}
	uses of $P_y$ and, with a probability larger than 99\%, outputs an $O(\epsilon)$-additive approximation of $\vec{a}$, which is defined as (\ref{eq:LSSolScale}).
\end{theorem}

\begin{proof}
The outline of the algorithm is as follows.
\begin{enumerate}
	\item Estimate the elements of $W$ and $\vec{z}$ in (\ref{eq:LSSolScale}) using a quantum algorithm based on QAE. Let the matrix and the vector consisting of the estimation results be $\hat{W}$ and $\vec{\hat{z}}$.
	\item Solve $\hat{W}\vec{a}=\vec{\hat{z}}$ by some classical solver of systems of linear equations (for example, row reduction). Let the solution be $\vec{\hat{a}}$. This is an output of our algorithm.
\end{enumerate}

Since we do not use the oracles $P_x,P_y$ in the second step, we focus on the step 1.
As we explained in Sec. \ref{sec:QAE}, we can obtain an estimation $\hat{w}_{ij}$ of $w_{ij}$ by QAE if we construct the following operators $A_{ij}$.
$A_{ij}$ transforms $\ket{0}$, a state in which all qubits are 0, to a state in the form of
\begin{equation}
\sqrt{w_{ij}}\ket{\psi} + \sqrt{1-w_{ij}}\ket{\psi_\bot}, \label{eq:stateforQAE}
\end{equation}
where $\ket{\psi}$ and $\ket{\psi_\bot}$ are some orthogonal states.
%some states and the second kets $\ket{0},\ket{1}$ correspond to single qubit states.
Such an operator is constructed as follows.
\begin{enumerate}
\renewcommand{\labelenumi}{(\roman{enumi})}
\item Prepare quantum registers $R_1,...,R_5$, which have enough qubits, and a single qubit register $R_6$. Set $R_1,R_2$ and the others to $\ket{i},\ket{j}$ and $\ket{0}$, respectively. \label{step1}
\item Create $\frac{1}{\sqrt{N_D}}\sum_{k=1}^{N_D}\ket{k}$, that is, an equiprobable superposition of $\ket{1},...,\ket{N_D}$ on $R_3$.
\item %Using $R_1$ as $i$ and $R_3$ as $k$, output $x^{(i)}_k$ on $R_4$ by $P_x$. Similarly, using $R_2$ as $j$ and $R_3$ as $k$, output $x^{(j)}_k$ on $R_5$.
Apply $P_x$ to a block of $R_1, R_3$ and $R_4$, which outputs $x_k^{(i)}$ on $R_4$. Similarly, apply $P_x$ to a block of $R_2, R_3$ and $R_5$, which outputs $x_k^{(j)}$ on $R_5$.
\item Using $x^{(i)}_k$ on $R_4$ and $x^{(j)}_k$ on $R_5$, transform $R_6$ from $\ket{0}$ to $\left(\sqrt{1-x^{(i)}_kx^{(j)}_k}\ket{0}+\sqrt{x^{(i)}_kx^{(j)}_k}\ket{1}\right)$ by some arithmetic circuits and controlled rotations\footnote{More concretely, we can perform the following calculation:
\begin{eqnarray}
	&&\ket{x^{(i)}_k}\ket{x^{(j)}_k}\ket{0}\ket{0}\ket{0} \nonumber \\
	&\rightarrow& \ket{x^{(i)}_k}\ket{x^{(j)}_k}\ket{x^{(i)}_kx^{(j)}_k}\ket{0}\ket{0} \nonumber \\
	&\rightarrow& \ket{x^{(i)}_k}\ket{x^{(j)}_k}\ket{x^{(i)}_kx^{(j)}_k}\Ket{\sqrt{x^{(i)}_kx^{(j)}_k}}\ket{0} \nonumber \\
	&\rightarrow& \ket{x^{(i)}_k}\ket{x^{(j)}_k}\ket{x^{(i)}_kx^{(j)}_k}\Ket{\sqrt{x^{(i)}_kx^{(j)}_k}}\left(\sqrt{1-x^{(i)}_kx^{(j)}_k}\ket{0}+\sqrt{x^{(i)}_kx^{(j)}_k}\ket{1}\right), \nonumber
\end{eqnarray}
where the first, second, and last kets correspond to $R_4$, $R_5$, and $R_6$, respectively, and the third and fourth kets correspond to two ancillary registers, which are not displayed in (\ref{eq:flowA}). We here used the multiplication circuit (e.g. \cite{MunozCoreas}) at the first arrow, the circuit for square root (e.g. \cite{MunozCoreas2}) at the second arrow, and the rotation gate on $R_6$ controlled by the second ancillary register at the last arrow. 
}. Then, the resultant state is in the form of (\ref{eq:stateforQAE}).% and the probability that we observe $1$ on $R_6$ in the state is $w_{ij}=\frac{1}{N_D}\sum_{k=1}^{N_D}x^{(i)}_kx^{(j)}_k$.
%\item Estimate the probability by QAE and let the result be $\hat{w}_{ij}$.
\end{enumerate}
Through the steps (i) to (iv), the state is transformed as follows.
\begin{eqnarray}
& & \ket{i}\ket{j}\ket{0}\ket{0}\ket{0}\ket{0} \nonumber \\
& \xrightarrow{{\rm (ii)}} & \ket{i}\ket{j}\left(\frac{1}{\sqrt{N_D}}\sum_{k=0}^{N_D}\ket{k}\right)\ket{0}\ket{0}\ket{0} \nonumber \\
& \xrightarrow{{\rm (iii)}} & \ket{i}\ket{j}\left(\frac{1}{\sqrt{N_D}}\sum_{k=0}^{N_D}\ket{k}\ket{x^{(i)}_k}\ket{x^{(j)}_k}\right)\ket{0} \nonumber \\
& \xrightarrow{{\rm (iv)}} & \ket{i}\ket{j}\left[\frac{1}{\sqrt{N_D}}\sum_{k=0}^{N_D}\ket{k}\ket{x^{(i)}_k}\ket{x^{(j)}_k}\right. \nonumber \\
& & \qquad \qquad \qquad \quad \otimes\left.\left(\sqrt{1-x^{(i)}_kx^{(j)}_k}\ket{0}+\sqrt{x^{(i)}_kx^{(j)}_k}\ket{1}\right)\right]  \nonumber \\
& & \label{eq:flowA}
\end{eqnarray}
In order to estimate $w_{ij}$ with an additive error $\epsilon^\prime$, QAE makes $O\left(\frac{1}{\epsilon^\prime}\right)$ uses of $A_{ij}$.
In the current case, $A_{ij}$ contains $O(1)$ uses of $P_x$ (precisely, two uses).
In total, we use it $O\left(\frac{1}{\epsilon^\prime}\right)$ times in the whole process of the QAE.

We can obtain an estimation $\hat{z}_i$ of $z_i$ in the similar way by replacing one of two $P_x$'s with $P_y$ and $x^{(j)}_k$ with $y_k$.

As shown in Lemma \ref{lem}, in order to obtain an $\epsilon$-additive approximation of $\vec{a}$, it is sufficient to take $\epsilon^\prime$ satisfying (\ref{eq:epsPrimeCond}).
Therefore, the number of calling $P_x,P_y$ in one QAE for $w_{ij}$ or $z_i$ is
\begin{equation}
O\left(\max\left\{\frac{d^{3/2}\kappa^4}{\epsilon},d\kappa^2\right\}\right). \label{eq:CompOneQAE}
\end{equation}
Here, we omit $c$ since it is independent of $N_D, \kappa$, $d$ (Assumption \ref{ass:scale}) and of course $\epsilon$.

Then, let us recall the issue of the success probability, which is mentioned in Sec. {\ref{sec:QAE}}.
Since the numbers of $w_{ij}$'s and $z_i$'s are $d(d+1)/2$ and $d$, respectively, in order to successfully estimate all of them with a probability larger than $0.99=1-0.01$, it suffices to make the success probability of each estimation larger than $1-\frac{0.01}{d^2}$.
Therefore, it is sufficient to repeat QAE $O\left(\log\left(\frac{d^2}{0.01}\right)\right)=O\left(\log(d)\right)$ times in estimating one of $w_{ij}$'s or $z_i$'s.
Combining this point and (\ref{eq:CompOneQAE}), and noting the numbers of $w_{ij}$'s and $z_i$'s, we see that the total number of calls to $P_x$ and $P_y$ in the whole algorithm are (\ref{eq:numPx}) and (\ref{eq:numPy}), respectively.

\end{proof}

\subsection{Time complexity}

Here, let us make some comments on overall {\it time complexity}.
Although we saw that we can reduce the {\it query complexity}, that is, the number of calls to $P_x$ and $P_y$, exponentially with respect to $N_D$ compared to the naive classical method, the time for one such query might depend on $N_D$.
Fortunately, in the following examples, the time complexity of one query is at most logarithmic with respect to $N_D$, which means that our method provides speed-up also in time complexity.

We first consider LSM, where the explanatory and objective variables are not given as exogenous data, but are produced by a simulation.
Concretely, $P_x$ is the oracle to generate sample paths of time evolution of asset prices, $P_y$ is the oracle to calculate the cashflows on paths, and $N_D$ is the number of the paths.
Typically, we generate asset price paths based on some stochastic differential equation, using random numbers.
The implementation of such a calculation as a quantum circuit has already been considered~\cite{Rebentrost,Stamatopoulos,Kaneko}.
In such a circuit, we create and use random numbers with $O(\log (N_D))$ digits in order to generate $N_D$ different paths.
The gate number for creating such numbers and arithmetic operations on them is $O({\rm polylog}(N_D))$, and so is the time complexity.
We can summarize this situation as follows: we aim to approximate some complex function by linear regression, and the quantum circuits which calculate the explanatory and objective variables with $O({\rm polylog}(N_D))$ time complexity are available.
We expect that this is not unique to LSM.

Next, let us consider the cases where we use a QRAM to implement the oracles.
According to \cite{Giovannetti}, once we are given the QRAM containing $N_D$ data points, we can load them into a register with $O({\rm polylog}(N_D))$ time complexity per query.
Thus, the cost of querying the oracle is negligible.
However, preparing such a QRAM also takes time, which must be taken into account.
The reasonable assumption is that we are initially given the data points as classical data, and load them into the QRAM, which takes $O(N_D)$ time complexity.
Even if we consider the QRAM initialization time, in practical uses of linear regression, our quantum method can provide a reduction of time complexity in the entire process of data analysis.
For example, consider the situation where we explore an optimal combination of preprocessing, such as one-hot encoding and logarithmic transformation, by trial and error.
Such transformations, by which we aim to enhance explanatory power of the variables, are ubiquitous in data analysis, and we often repeat regression for one data set applying various patterns of transformation and compare the results.
In this situation, the costliest part of the total calculation is a series of classical regression in many trials for the classical method, but a single QRAM initialization at the very beginning for the our method, which will be less time-consuming.
Therefore, we conclude that the proposed method is totally faster than the naive classical method.

\section{\label{sec:extension}Extention to a class of convex optimization}

\subsection{Linear regression as optimization by Newton's method}

In this section, we extend our method for linear regression to more general optimization problems.
%We here discuss the relationship between our method and previously proposed ones.
Before this, we firstly present another interpretation of (\ref{eq:LSSol}), the formula for the solution of linear regression.
We regard linear regression as an optimization problem of
\begin{equation}
F(\vec{a}) = \frac{1}{2N_D}\| \vec{y}-X\vec{a}\|^2 = \frac{1}{2N_D}\sum_{k=1}^{N_D} \left(\sum_{i=1}^da_ix^{(i)}_k-y_k\right)^2. \label{eq:objFuncLS}
\end{equation}
This can be rewritten as
\begin{equation}
F(\vec{a}) = \frac{1}{N_D}\sum_{k=1}^{N_D} f(\vec{a};\vec{x}_k,y_k),
\end{equation}
where
\begin{equation}
f(\vec{a};\vec{x}_k,y_k) = \frac{1}{2}\left(\sum_{i=1}^da_ix^{(i)}_k-y_k\right)^2.
\end{equation}
The first and second derivatives of $F$ are
\begin{equation}
\frac{\partial}{\partial a_i}F(\vec{a})  = \frac{1}{N_D} \sum_{k=1}^{N_D} \frac{\partial}{\partial a_i}f(\vec{a};\vec{x}_k,y_k)= \frac{1}{N_D}\sum_{k=1}^{N_D} x_k^{(i)}\left(\sum_{j=1}^{d}a_jx_k^{(j)}-y_k\right),
\end{equation}
\begin{equation}
\frac{\partial^2}{\partial a_i \partial a_j}F(\vec{a}) = \frac{1}{N_D}\sum_{k=1}^{N_D} \frac{\partial^2}{\partial a_i \partial a_j}f(\vec{a};\vec{x}_k,y_k)=\frac{1}{N_D}\sum_{k=1}^{N_D} x_k^{(i)}x_k^{(j)},
\end{equation}
respectively.
This means that, $W=\frac{1}{N_D}X^TX$ and $-\vec{z} = -\frac{1}{N_D}X^T\vec{y}$ are the Hessian matrix and the gradient vector of $F$ at $\vec{a}=\vec{0}$, respectively.
Besides, the updating formula in Newton's method is
\begin{equation}
\vec{a}_{n+1} = \vec{a}_n-H_F^{-1}(\vec{a_n})\vec{g}_F(\vec{a}_n), \label{eq:Newton}
\end{equation}
where $H_F$ and $\vec{g}_F$ are the Hessian and the gradient of the function $F$, respectively, and $\vec{a}_n$ is the optimization variable after the $n$-th update.
Then, we can interpret (\ref{eq:LSSol}) as an one-time update in Newton's method from the initial point $\vec{a}_0=\vec{0}$.
Note that Newton's method gives the exact solution by only one update from any initial point, if the objective function is quadratic.

In summary, we can consider our method as optimization of the objective function (\ref{eq:objFuncLS}) by Newton's method, where calculation of the gradient and the Hessian, which is time-consuming in the classical method, is done by the QAE-based method.

\subsection{Extension of our method : the QAE-based Newton's method}

%Obviously, we can extend our method, that is, Newton's method with a gradient and a Hessian estimated by the QAE-based method, to the more general class of optimization problems, including the setting in \cite{Shao,Liu}.
On the basis of the above discussion, it is now straightforward to extend our method to a more general class of optimization problems.
That is, keeping the updating formula (\ref{eq:Newton}), we can perform Newton's method based on the gradient and the Hessian estimated by QAE.

Concretely, we consider an optimization problem in which the objective function can be written as a sum of the values of some function with different inputs:
\begin{equation}
F(\vec{a}) = \frac{1}{N_D}\sum_{k=1}^{N_D} f(\vec{a},\vec{c}_k). \label{eq:objFunc}
\end{equation}
Here, $f$ is the real-valued twice-differentiable function, which are shared by all the terms.
Its inputs are the optimization variables $\vec{a}\in \mathbb{R}^{d}$ and some parameters $\vec{c}_k$, which are different in each term.
Some conditions on $F$ necessary for convergence analysis are given in Sec. \ref{sec:convAnal}.
It is obvious that the objective functions (\ref{eq:objFuncLS}) fall into (\ref{eq:objFunc}).
For the objective function like (\ref{eq:objFunc}), the gradient $\vec{g}_F(\vec{a})=(g_{F,1}(\vec{a}),...,g_{F,d}(\vec{a}))^T$ and the Hessian $H_F(\vec{a})=(h_{F,ij}(\vec{a}))_{\substack{1\le i \le d \\ 1\le j \le d}}$ are given as
\begin{equation}
g_{F,i}(\vec{a}) = \frac{\partial}{\partial a_i}F(\vec{a}) = \frac{1}{N_D}\sum_{k=1}^{N_D} f_i(\vec{a},\vec{c}_k),
\end{equation}
\begin{equation}
h_{F,ij}(\vec{a}) = \frac{\partial^2}{\partial a_i\partial a_i}F(\vec{a}) = \frac{1}{N_D}\sum_{k=1}^{N_D} f_{ij}(\vec{a},\vec{c}_k),
\end{equation}
where we simply write $\frac{\partial}{\partial a_i}f$ and $\frac{\partial^2}{\partial a_i \partial a_j}f$ as $f_i$ and $f_{ij}$, respectively.

Then, we estimate the gradient and the Hessian as follows.
We assume the availability of the followings:
\begin{itemize}
	\item $P_c$, which outputs $\vec{c}_k$ for given $k$:
	\begin{equation}
	P_c:\ket{k}\ket{0}\mapsto \ket{k}\ket{\vec{c}_k}.
	\end{equation}
	This can be constructed by QRAM.
	
	\item $P_{i}$ for $i=1,..,d$, which outputs $f_i(\vec{a},\vec{c}_k)$ for given $\vec{a}$ and $\vec{c}_k$:
	\begin{equation}
	P_{i}:\ket{\vec{a}}\ket{\vec{c}_k}\ket{0} \mapsto \ket{\vec{a}}\ket{\vec{c}_k}\ket{f_i(\vec{a},\vec{c}_k)}.
	\end{equation}
	
	\item $P_{ij}$ for $i,j=1,...,d$, which outputs $f_{ij}(\vec{a},\vec{c}_k)$ for given $\vec{a}$ and $\vec{c}_k$:
	\begin{equation}
	P_{ij}:\ket{\vec{a}}\ket{\vec{c}_k}\ket{0} \mapsto \ket{\vec{a}}\ket{\vec{c}_k}\ket{f_{ij}(\vec{a},\vec{c}_k)}.
	\end{equation}
\end{itemize}
Then, preparing appropriate registers, we can perform the following computation:
\begin{eqnarray}
& & \ket{\vec{a}}\ket{0}\ket{0}\ket{0}\ket{0} \nonumber \\
& \rightarrow & \left(\frac{1}{\sqrt{N_D}}\sum_{k=0}^{N_D}\ket{\vec{a}}\ket{k}\right)\ket{0}\ket{0}\ket{0} \nonumber \\
& \rightarrow & \left(\frac{1}{\sqrt{N_D}}\sum_{k=0}^{N_D}\ket{\vec{a}}\ket{k}\ket{\vec{c}_k}\right)\ket{0}\ket{0} \nonumber \\
& \rightarrow & \left(\frac{1}{\sqrt{N_D}}\sum_{k=0}^{N_D}\ket{\vec{a}}\ket{k}\ket{\vec{c}_k}\Ket{f_i(\vec{a};\vec{c}_k)}\right)\ket{0} \nonumber \\
& \rightarrow & \frac{1}{\sqrt{N_D}}\sum_{k=0}^{N_D}\ket{\vec{a}}\ket{k}\ket{\vec{c}_k}\Ket{f_i(\vec{a};\vec{c}_k)} \nonumber \\
& & \qquad \qquad \quad \otimes \left(\sqrt{1-f_i(\vec{a};\vec{c}_k)}\ket{0}+\sqrt{f_i(\vec{a};\vec{c}_k)}\ket{1}\right), \label{eq:flowLDeriv}
\end{eqnarray}
where $P_c$ and $P_i$ are used at the second and third arrows, respectively.
We then obtain an estimation of $g_i(\vec{a})$ by estimating the probability that the last qubit takes 1 by QAE.
We can also estimate $h_{ij}(\vec{a})$ similarly, replacing $P_i$ with $P_{ij}$ and $f_i(\vec{a},\vec{c}_k)$ with $f_{ij}(\vec{a},\vec{c}_k)$.
Again, in order to estimate one $g_i(\vec{a})$ or $h_{ij}(\vec{a})$ with $\epsilon$-additive error, the number of calling $P_i$ and $P_{ij}$ is at most $O(1/\epsilon)$, which means the exponential speedup with respect to $N_D$ compared with classical iterative calculation.

Using the gradient and the Hessian estimated as above, we update $\vec{a}_n$ to $\vec{a}_{n+1}$ similarly to (\ref{eq:Newton}), that is,
\begin{equation}
\vec{a}_{n+1} = \vec{a}_n-\hat{H}_F^{-1}(\vec{a_n})\vec{\hat{g}}_F(\vec{a}_n), \label{eq:NewtonErr}
\end{equation}
where $\vec{\hat{g}}_F(\vec{a_n})$ and $\hat{H}_F(\vec{a_n})$ are the estimated gradient and Hessian, respectively.
After sufficiently many iterations, we obtain the approximated solution of the optimization.
Hereafter, we call this method the {\it QAE-based Newton's method}.

Note that $\hat{H}_F(\vec{a_n})$ must be invertible so that the update (\ref{eq:NewtonErr}) can be defined.
Hereafter, we consider the situation where the original Hessian $H_F(\vec{a_n})$ is positive-definite and therefore invertible.
In order to keep such a property, we have to obtain $\hat{H}_F(\vec{a_n})$ accurately enough.

\subsection{Convergence analysis of the QAE-based Newton's method\label{sec:convAnal}}

Let us estimate the complexity of the QAE-based Newton's method to obtain an approximated solution with $\epsilon$-additive error.
For a mathematically rigorous discussion, we first make some assumptions on the objective function $F:\mathbb{R}^d\rightarrow\mathbb{R}$.

\begin{assumption}
	$F$ is twice-differentiable. \label{ass:twiceDiff}
\end{assumption}
\noindent This is reproduced since we assumed that $f$ in (\ref{eq:objFunc}) is twice-differentiable.

\begin{assumption}
	$F$ is $\mu$-strongly convex, that is, there exists a positive number $\mu$ such that
	\begin{equation}
	\forall \vec{a},\vec{b}\in \mathbb{R}^d, F(\vec{a})\ge F(\vec{b}) + \vec{g}_F(\vec{b})\cdot(\vec{a}-\vec{b})+\frac{\mu}{2}\|\vec{a}-\vec{b}\|^2.
	\end{equation}
	\label{ass:StConv}
\end{assumption}
\noindent This assumption means that the eigenvalues of $H_F(\vec{a})$ are greater than or equal to $\mu$ for any $\vec{a}\in\mathbb{R}^d$, that is,
\begin{equation}
\forall \vec{a}\in\mathbb{R}^d, H_F(\vec{a}) \succeq \mu I_d, \label{eq:HEigenLB}
\end{equation}
where $I_d$ is the $d\times d$ identity matrix and for $A,B\in \mathbb{R}^d$, $A\succeq B$ means that $A-B$ is positive-semidefinite.
This immediately leads to
\begin{equation}
\forall a\in \mathbb{R}^d, \|(H_F(\vec{a}))^{-1}\|\le \frac{1}{\mu}. \label{eq:HInvEigenUB}
\end{equation}

\begin{assumption}
	$F$ has an $M$-Lipschitz Hessian, that is, 
	\begin{equation}
	\forall \vec{a},\vec{b}\in \mathbb{R}^d, \|H_F(\vec{a}) - H_F(\vec{b})\| \le M\|\vec{a}-\vec{b}\|.
	\end{equation}
	\label{ass:LipHess}
\end{assumption}
\noindent This assumption leads to the following inequality:
\begin{equation}
\forall \vec{a},\vec{b}\in \mathbb{R}^d, \|\vec{g}_F(\vec{a})-\vec{g}_F(\vec{b}) - H_F(\vec{b})(\vec{a}-\vec{b})\| \le \frac{M}{2}\|\vec{a}-\vec{b}\|^2. \label{eq:gDiffBound}
\end{equation}

Assumptions \ref{ass:twiceDiff} to \ref{ass:LipHess} are usually made in the discussion on convergence properties of the {\it ordinary} Newton's method, that is, cases where the gradient and the Hessian can be exactly computed in the classical way (see, for example, \cite{Nesterov}).
We do not make any additional assumptions for the QAE-based Newton's method.

Next, let us define some quantities.
Since QAE introduces errors in the estimated Hessian and gradient, we have to consider Newton's method in which the update difference contains some error.
For $\vec{a}\in\mathbb{R}^d$, we define $\vec{\Delta}_F(\vec{a})$ as
\begin{equation}
\vec{\Delta}_F(\vec{a}) := (\hat{H}_F(\vec{a}))^{-1}\vec{\hat{g}}_F(\vec{a})-(H_F(\vec{a}))^{-1}\vec{g}_F(\vec{a}). \label{eq:HgErr}
\end{equation}
Besides, we write the minimum\footnote{Since we are considering the convex optimization as stated in Assumption \ref{ass:StConv}, there is only one global minimum.} that we search as $\vec{a}^\star:=\argmin_{\vec{a}\in\mathbb{R}^d}F(\vec{a})$, the optimization variable after the $n$-th update as $\vec{a}_n$, and the difference between $\vec{a}^\star$ and $\vec{a}_n$ as $\delta_n:=\|\vec{a}_n-\vec{a}^\star\|$.

Then, we can show the following lemma, which is repeatedly used.

\begin{lemma}
	Let Assumptions \ref{ass:twiceDiff} to \ref{ass:LipHess} be satisfied.
	Then, for any non-negative number $\epsilon$ and any $\vec{a}\in\mathbb{R}^d$ such that
	\begin{equation}
		\|\vec{\Delta}_{F}(\vec{a})\| \le \epsilon, \label{eq:ErrOrd}
	\end{equation}
	the following inequality holds:
	\begin{equation}
	\delta^\prime \le \frac{M}{2\mu}\delta^2 + \epsilon, \label{eq:QAENErrConv}
	\end{equation}
	where $\delta:=\|\vec{a}-\vec{a}^\star\|$, $\delta^\prime:=\|\vec{a}^\prime-\vec{a}^\star\|$, and $\vec{a}^\prime=\vec{a}-(\hat{H}_F(\vec{a}))^{-1}\vec{\hat{g}}_F(\vec{a})$.
	Furthermore, when
	\begin{equation}
	\frac{2M}{\mu}\epsilon<1, \label{eq:epsCond}
	\end{equation}
	is satisfied, the following hold:
	\begin{eqnarray}
		&\delta^\prime < \delta & ; \ {\rm if} \ \delta_-<\delta<\delta_+ \\
		&\delta^\prime \le \delta_- & ; \ {\rm if} \ \delta\le\delta_-
	,\label{eq:delpr}
	\end{eqnarray}
	where  
	\begin{equation}
	\delta_\pm:=\frac{\mu}{M}\left(1\pm\sqrt{1-\frac{2M}{\mu}\epsilon}\right). \label{eq:deltapm}
	\end{equation}
	\label{lem:deltarel}
\end{lemma}

The proof is given in Appendix \ref{sec:app2}.
Here, let us comment on what Lemma \ref{lem:deltarel} implies. 
%Equation (\ref{eq:QAENErrConv}) indicates the following: when the update differences in Newton's method contain errors at most $\epsilon$, the difference $\delta$ between the optimization variables $\vec{a}$ and the optimal point $\vec{a}^{\star}$ quadratically converges like Newton's method with no error, as long as $\epsilon$ is much smaller than $\frac{M}{2\mu}\delta^2$, or, more strictly, $\delta_-<\delta$.
Equation (\ref{eq:QAENErrConv}) indicates that, even when the update differences in Newton's method contain errors at most $\epsilon$, the difference $\delta$ between the optimization variables $\vec{a}$ and the optimal point $\vec{a}^{\star}$ quadratically converges like Newton's method with no error, as long as $\epsilon \lesssim \frac{M}{2\mu}\delta^2$.
More strictly, while $\delta_-<\delta$, $\delta$ decreases at every update, as shown in (\ref{eq:delpr}).
On the other hand, after $\delta$ reaches $\delta_-$, $\epsilon$ is not negligible in (\ref{eq:QAENErrConv}), and therefore $\delta$ does not necessarily decreases.
Nevertheless, $\delta$ does not exceeds $\delta_-$, once it goes below.
Since $\delta_-<2\epsilon$, we can make $\vec{a}$ converge with desired accuracy if we can suppress $\epsilon$.

Using Lemma \ref{lem:deltarel}, we obtain the following lemma, which shows how many updates are sufficient for $\delta$ to reach $2\epsilon$ in Newton's method with erroneous update differences.

\begin{lemma}
	Let Assumptions \ref{ass:twiceDiff} to \ref{ass:LipHess} be satisfied.
	Suppose that we repeatedly updates $\vec{a}\in\mathbb{R}^d$ by (\ref{eq:NewtonErr}) from some initial point $\vec{a}_0$, where we write the result of $n$-times updates as $\vec{a}_n$ and define $\delta_n:=\|\vec{a}_n-\vec{a}^\star\|$ for $n=0,1,2,...$. 
	Then, for any positive number $\epsilon$ satisfying (\ref{eq:epsCond}) and any $\vec{a}_0$ such that
	\begin{equation}
	\|\vec{a}_0-\vec{a}^\star\|<\frac{\mu}{M}, \label{eq:iniCond}
	\end{equation}
 	if (\ref{eq:ErrOrd}) holds for $\vec{a}=\vec{a}_n,n=0,1,2,...$, there exists a non-negative integer $n_{\rm it}$ such that
	\begin{equation}
	\delta_n\le 2\epsilon \label{eq:tgtErr}
	\end{equation}
	for any $n\ge n_{\rm it}$, where 
	\begin{equation}
	n_{\rm it} := \max \left\{\left\lceil\log_2 \left(\frac{\log\left(\frac{2M\epsilon}{\mu}\right)}{2\log\left(\frac{M\delta_0}{\mu}\right)}\right)\right\rceil + 1,1\right\}, \label{eq:iterNum}
	\end{equation}
	\label{lem:ErrNewConv}
\end{lemma}

The proof is given in Appendix \ref{sec:app3}

Next, we consider how accurate we should estimate the gradient and the Hessian in order to suppress the error in the update difference to the order of $\epsilon$.

\begin{lemma}
	Let Assumptions \ref{ass:twiceDiff} to \ref{ass:LipHess} be satisfied.
	Besides, suppose that we are given a positive integer $n$, a positive number $\epsilon$ satisfying (\ref{eq:epsCond}), and $\vec{a}\in \mathbb{R}^d$ satisfying $\delta<\frac{2\mu}{M}$, where $\delta:=\|\vec{a}-\vec{a}^\star\|$. 
	Then, in order for $\Delta_F(\vec{a})$ defined as (\ref{eq:HgErr}) to satisfy $\|\vec{\Delta}_{F}(\vec{a})\| \le \epsilon$, it is sufficient that each component of the estimated gradient $\vec{\hat{g}}_F(\vec{a})$ and Hessian $\hat{H}_F(\vec{a})$ has an additive error $\epsilon^\prime_g$ and $\epsilon^\prime_H$ such that
	\begin{equation}
	\epsilon^\prime_g \le \frac{\mu\epsilon}{2d^{1/2}}. \label{eq:epg}
	\end{equation}
	and
	\begin{equation}
	\epsilon^\prime_H\le\frac{\mu\epsilon}{4\tilde{\delta} d}, \label{eq:epH}
	\end{equation}
	where $\tilde{\delta}:=\max\{\delta,\delta_-\}$, respectively.
	\label{lem:comperr}
\end{lemma}

\begin{lemma}
	Let Assumptions \ref{ass:twiceDiff} to \ref{ass:LipHess} be satisfied.
	Besides, suppose that we are given a positive number $\epsilon$ satisfying (\ref{eq:epsCond}) and an initial point $\vec{a}_0\in \mathbb{R}^d$ satisfying $\delta_0:=\|\vec{a}_0-\vec{a}^\star\|<\delta_+$.
	Furthermore, let $\vec{a}_n\in\mathbb{R}^d$ be a vector given by $n$-times updates by (\ref{eq:NewtonErr}) from the initial point $\vec{a}_0$.
	Then, $\Delta_F(\vec{a}_n)$ defined as (\ref{eq:HgErr}) satisfies
	\begin{equation}
	\|\Delta_F(\vec{a}_n)\|\le\epsilon \label{eq:ErrOrd2}
	\end{equation}
	for any positive integer $n$, if each component of $\vec{\hat{g}}_F(\vec{a})$ and $\hat{H}_F(\vec{a})$ has an additive error $\epsilon^\prime_g$ and $\epsilon^\prime_H$ satisfying (\ref{eq:epg}) and
	\begin{equation}
	\epsilon^\prime_H\le\frac{\mu\epsilon}{4\tilde{\delta}_0 d}, \label{eq:epH2}
	\end{equation}
	where $\tilde{\delta}_0:=\max\{\delta_0,\delta_-\}$, respectively.
	\label{lem:comperr2}
\end{lemma}

The proofs of Lemmas \ref{lem:comperr} and \ref{lem:comperr2} are given in Appendix \ref{sec:app4} and \ref{sec:app5}, respectively.

Combining Lemma \ref{lem:ErrNewConv} and Lemma \ref{lem:comperr2}, we immediately obtain the following theorem.
\begin{theorem}
	Let Assumptions \ref{ass:twiceDiff} to \ref{ass:LipHess} be satisfied.
	Besides, suppose that we are given a positive number $\epsilon$ satisfying (\ref{eq:epsCond}) and $\vec{a}_0$ satisfying (\ref{eq:iniCond}).
	Then, using the QAE-based Newton's method which is based on the updating formula (\ref{eq:NewtonErr}), we obtain a $2\epsilon$-additive approximation of $\vec{a}^\star$ by $n_{\rm it}$-times updates, where $n_{\rm it}$ is given by (\ref{eq:iterNum}), with a success probability higher than 99\%.
	In the process, the total number of calls $P_i,i=1,...,d$ is
	\begin{equation}
	N_{\rm 1stDer} = O\left(\frac{d^{3/2}}{\mu \epsilon}n_{\rm it}\log(n_{\rm it}d^2)\right), \label{eq:N1st}
	\end{equation}
	that for $P_{ij},i,j=1,...,d$ is
	\begin{equation}
	N_{\rm 2ndDer} = O\left(\frac{\tilde{\delta}_0 d^3}{\mu \epsilon}n_{\rm it}\log(n_{\rm it}d^2)\right), \label{eq:N2nd}
	\end{equation}
	and that for $P_c$ is
	\begin{equation}
	N_c=N_{\rm 1stDer}+N_{\rm 2ndDer}. \label{eq:Nc}
	\end{equation}
	\label{theoGen}
\end{theorem}

\begin{proof}
At First, let us consider a situation where all the estimation processes in the QAE-based Newton's method are successful, more concretely, all the estimations have errors smaller than the desired level.
Correspondingly, suppose that, for $\epsilon^\prime_g$ satisfying (\ref{eq:epg}) and $\epsilon^\prime_H$ satisfying (\ref{eq:epH2}), we can calculate each component of $\vec{\hat{g}}_F$ and $\hat{H}_F$ with an $\epsilon^\prime_g$-additive error and an $\epsilon^\prime_H$-additive error, respectively, at each update in the QAE-based Newton's method.
Because of Lemma \ref{lem:comperr2}, this means that (\ref{eq:ErrOrd2}) is satisfied at each update.
Then, from Lemma \ref{lem:ErrNewConv}, we obtain the optimization variable $\vec{a}_{n_{\rm it}}$ satisfying (\ref{eq:tgtErr}), which is therefore $2\epsilon$-approximation of $\vec{a}^\star$, after $n_{\rm it}$-time updates.

Then, let us evaluate the total number of calls to $P_{i}$ and $P_{i,j}$ in the above process, recalling that the success probability of QAE is not 1.
We calculate the elements of $\vec{\hat{g}}_F$ and $\hat{H}_F$, whose number is $O(d^2)$ in total, at $n_{\rm it}$ updates.
Therefore, the total number of such calculations is $O(n_{\rm it}d^2)$.
To make the probability that all calculations are done with required errors larger than 99\%, it suffices to enhance the success probability of each calculation to $1-\frac{0.01}{n_{\rm it}d^2}$.
This can be achieved by repeating QAEs $O\left(\log (n_{\rm it}d^2)\right)$ times in each calculation.
Besides, in one trial in repeating QAEs, it is sufficient to make $O\left(\frac{1}{\epsilon^\prime_g}\right)$ calls to $P_{i}$ (resp. $O\left(\frac{1}{\epsilon^\prime_H}\right)$ calls to $P_{i,j}$) in order to get an element of $\vec{\hat{g}}_F$ (resp. $\hat{H}_F$) with an $\epsilon^\prime_g$-additive (resp. $\epsilon^\prime_H$-additive) error.
In summary, we calculate the $d$ components of $\vec{\hat{g}}_F$ by $O\left(\log (n_{\rm it}d^2)\right)$-time repeating QAEs with $O\left(\frac{1}{\epsilon^\prime_g}\right)$ calls to $P_{i}$ at each of $n_{\rm it}$ updates, and therefore the total call number of $P_{i}$ becomes (\ref{eq:N1st}).
Similarly, that of $P_{i,j}$ becomes (\ref{eq:N2nd}).
 
Since one call to $P_i$ or $P_{i,j}$ accompanies one call to $P_c$, we obtain (\ref{eq:Nc}).
\end{proof}

\subsection{The Newton's method based on classical Monte Carlo}

The important feature of the QAE-based Newton's method is that we estimate the components in the Hessian and the gradient, which are represented as the sum of many similar terms, by QAE, or, more concretely, the procedure similar to the quantum algorithm for Monte Carlo integration based on QAE~\cite{Montanaro,Suzuki}.
%In fact, the quantum method for Monte Carlo can be replaced with classical Monte Carlo.
It is natural to compare the QAE-based Newton's method with the classical Monte Carlo (CMC)-based Newton's method, which is just Newton's method whose summation parts for calculating the Hessian and the gradient are replaced with CMC.
If we can use the oracles similar to those in the above discussion, $P_c, P_i$ and $P_{i,j}$ in Sec. \ref{sec:extension} on a classical computer, we can estimate the components of the Hessian and the gradient also by CMC, that is, we can take the average of the values of sampled terms in the sum.
Note that the classical counterpart of a QRAM, which can be used for $P_x, P_y$ and $P_c$, is just a RAM.

Let us estimate the complexity of this method.
In the classical Monte Carlo, the estimation error is proportional to $N_{\rm samp}^{-1/2}$, where $N_{\rm samp}$ is the sample number.
More precisely speaking, asymptotically, the deviation of the estimated value from the true value is less than $\frac{C}{\sqrt{N_{\rm samp}}}$ with a probability $1-\gamma$, where $C=O({\rm polylog}(\gamma))$.
In other words, it is sufficient to take $O\left(\frac{{\rm polylog} (\gamma)}{\epsilon^2}\right)$ samples in order to obtain an estimation with error less than $\epsilon$ with probability higher than $1-\gamma$.
Besides, note that the number of calling for each oracle is proportional to the sample number.
Then, by a discussion similar to that on Theorem \ref{theoGen}, we see that we obtain a $2\epsilon$-additive approximation of $\vec{a}^\star$ by the CMC-based Newton's method with a probability higher than 99\%.
In the process, the iteration number $n_{\rm it}$ is given as (\ref{eq:iterNum}), the total number of calling $P_i,i=1,...,d$ is
\begin{equation}
	N_{\rm 1stDer} = O\left(\frac{d^2}{\mu^2 \epsilon^2}n_{\rm it}{\rm polylog}(n_{\rm it}d^2)\right),
\end{equation}
that for $P_{ij},i,j=1,...,d$ is
\begin{equation}
	N_{\rm 2ndDer} = O\left(\frac{(\tilde{\delta}_0)^2 d^4}{\mu^2 \epsilon^2}n_{\rm it}{\rm polylog}(n_{\rm it}d^2)\right),
\end{equation}
and that for $P_c$ is (\ref{eq:Nc}).
Therefore, as the QAE-based Newton's method does, the CMC-based Newton's method removes the dependency of the query complexity on $N_D$, even though it is classical.
Note that the complexity of the QAE-based Newton's method is smaller than the CMC-based one, due to the so-called quadratic quantum speedup of Monte Carlo integration.

\subsection{Related previous works}

As a final comment, we briefly refer to some previous works \cite{Shao2,Liu}, which used similar ideas to ours in the sense that the derivatives of the objective function are calculated by QAE.
These papers considered not Newton's method but gradient descent and its application to an optimization problem with an objective function consisting of many similar terms, which is less general than (\ref{eq:objFunc}) but includes linear regression.
Then, they aimed at speedup of gradient calculation by QAE similarly to ours.
Compared with Newton's method, gradient descent has both pros and cons.
As pros, it requires not a Hessian but only a gradient, and it reaches the minimum from any initial point in convex optimization.
On the other hand, as cons, it usually requires more iterations than Newton's method and there is an issue on choice of appropriate step size.
%Especially, linear regression, for which it is guaranteed that Newton's method finds the solution by one-time update from any initial point, can be a good use case of the QAE-based Newton's method.

\section{\label{sec:summary}Summary}

In this paper, we proposed a quantum algorithm for linear regression, or, more concretely, estimation of regression coefficients as classical data.
Existing algorithms such as \cite{Wiebe,Wang} create the quantum state encoding coefficients in its amplitude by the HHL algorithm \cite{HHL} or its modification and then read out the coefficients.
On the other hand, in our method, we estimate the elements of $W=\frac{1}{N_D}X^TX$ and $\vec{z}=\frac{1}{N_D}X^T\vec{y}$, where $X=(x^{(i)}_k)_{\substack{1\le i \le d \\ 1 \le k \le N_D}}$ is the design matrix, $\vec{y}=(y_1,...,y_k)^T$ is the objective variable vector, and $N_D$ is the number of the data points, and then we find the coefficients by classical computation of $\vec{a}=W^{-1}\vec{z}$.
Since, as shown in (\ref{eq:w}) and (\ref{eq:z}), the elements have the form of the sum of $x^{(i)}_kx^{(j)}_k$ or $x^{(i)}_ky_k$ over data points, we can estimate them by QAE~\cite{Brassard,Suzuki,Aaronson,Grinko,Nakaji,Brown,Tanaka}, assuming availability of the oracles $P_x$ and $P_y$ which output $x^{(i)}_k$ and $y_k$, respectively, for specified $i$ and $k$.
The query complexity of our method is given as (\ref{eq:numPx}) and (\ref{eq:numPy}), which means exponential speedup with respect to $N_D$ compared with the naive classical method, and improvement with respect to the tolerance error $\epsilon$ compared with the previous quantum methods such as \cite{Wiebe,Wang}.

Finally, we extended our method to more general optimization problems, that is, convex optimization with an objective function consisting of many similar terms like (\ref{eq:objFunc}).
In light of the interpretation of linear regression as Newton's method, we proposed the QAE-based Newton's method, in which the gradient and the Hessian are estimated by QAE.
Introducing effects of estimation errors in the ordinary discussion on convergence of Newton's method, we derived the convergence property and the query complexity of the QAE-based Newton's method.
Even if there are estimation errors, the method shows well-known quadratic convergence and reaches the solution in a small number of iterations.
%Especially, in linear regression, we can obtain the solution with one-time iteration from any initial point, and therefore it is a good use case of the QAE-based Newton's method.

Obtaining the improved dependence of complexity on $\epsilon$, we expect that we can apply our method to mid-sized but many-times-repeated regression problems like LSM.%problems like LSM, which is mid-sized with respect to data point number but desired to be sped up.
Generally speaking, our method can be better than the previous quantum methods and the naive classical method when $d,\kappa \ll \frac{1}{\epsilon} \ll N_D$.
On the other hand, since the complexity of our method depends on $d$ and $\kappa$ more strongly than the previous quantum methods, they will be better when $\frac{1}{\epsilon} \ll d$ or $\frac{1}{\epsilon} \ll \kappa$.
Since the naive classical method does not induce any error as quantum ones, it will be better when $N_D \ll \frac{1}{\epsilon}$.

In future work, we will consider implementation of LSM on quantum computers using our method for linear regression, as a concrete and practical use case of quantum computing in financial industry.

\section*{Acknowledgment}

This work was supported by MEXT Quantum Leap Flagship Program (MEXT Q-LEAP) Grant Number JPMXS0120319794.

\appendix

\section{Proofs of the lemmas}

\subsection{\label{sec:app} Proof of Lemma \ref{lem}}

\begin{proof}
	First, let us prove that $\hat{W}$ is invertible under
	\begin{equation}
	\epsilon^\prime < c/d\kappa^2. \label{eq:condHatWInv}
	\end{equation}
	Note that $W$ is positive-definite, and, according to discussions on perturbation on eigenvalues (see, for example, \cite{Trefethen}),
	\begin{equation}
		|\hat{\lambda}_{\rm min} - \lambda_{\rm min}| \le \|\delta W\|, \label{eq:WWhatEigenRel}
	\end{equation}
	where $\lambda_{\rm min}$ and $\hat{\lambda}_{\rm min}$ are the smallest eigenvalues of $W$ and $\hat{W}$, respectively.
	Therefore, it is sufficient to prove $\|\delta W\|< \lambda_{\rm min}$ in order to prove that $\hat{W}$ is positive definite, and then invertible.
	
	Then, let us bound $\lambda_{\rm min}$ first.
	Note that
	\begin{equation}
		\|W\| > c. \label{eq:ineq1}
	\end{equation}
	We can see this from
	\begin{equation}
		\|W\| \geq \frac{1}{d}{\rm Tr} W
	\end{equation}
	and
	\begin{equation}
		{\rm Tr} W = \sum_{i=1}^{d}w_{ii}>dc,
	\end{equation}
	where we used Assumption \ref{ass:scale}.
	(\ref{eq:ineq1}) leads to the following inequality:
	\begin{equation}
		\lambda_{\rm min} > \frac{c}{\kappa^2}, \label{eq:lambdaBound}
	\end{equation}
	since
	\begin{equation}
		\lambda_{\rm min} = \frac{\lambda_{\rm max}}{\kappa(W)} = \frac{\|W\|}{\kappa^2},
	\end{equation}
	where $\lambda_{\rm max}$ is the largest eigenvalue of $W$ and we used $\kappa(W)=\kappa^2=\lambda_{\rm max}/\lambda_{\rm min}$ and $\|W\|=\lambda_{\rm max}$. 
	
	Next, let us consider $\|\delta W\|$.
	This can be upper bounded as follows:
	\begin{equation}
		\|\delta W\| < d\epsilon^\prime, \label{eq:ineq3}
	\end{equation}
	This is because
	\begin{equation}
		\|\delta W\| \le \|\delta W\|_F = \left(\sum_{i=1,...,d}\sum_{j=1,...,d}|\delta w_{ij}|^2\right)^{1/2} < d\epsilon^\prime,
	\end{equation}
	where the last inequality is satisfied due to the assumption $|\delta w_{ij}|<\epsilon^\prime$.
	
	Combining (\ref{eq:lambdaBound}) and (\ref{eq:ineq3}), we see that, for $\epsilon^\prime$ satisfying (\ref{eq:condHatWInv}), $\|\delta W\|< \lambda_{\rm min}$ holds, and therefore $\hat{W}$ is invertible.
	
	\ \\
	
	Then, let us prove (\ref{eq:erreval}).
	We use a well-known formula (see \cite{Golub})
	\begin{equation}
	\frac{\|\vec{\hat{a}}-\vec{a}\|}{\|\vec{a}\|} \le \kappa(W) \left(\frac{\|\delta W\|}{\|W\|} + \frac{\|\delta\vec{z}\|}{\|\vec{z}\|}\right) + O((\epsilon^\prime)^2).
	\end{equation}
	This can be transformed as  
	\begin{equation}
	\|\vec{\hat{a}}-\vec{a}\| \le \kappa^4 \left(\frac{\|\vec{z}\|\cdot\|\delta W\|}{\|W\|^2} + \frac{\|\delta\vec{z}\|}{\|W\|}\right) + O((\epsilon^\prime)^2),  \label{eq:deltaabound}
	\end{equation}
	where we used
	\begin{equation}
	\|\vec{a}\|=\|W^{-1}\vec{z}\| \le \|W^{-1}\| \cdot \|\vec{z}\| = \frac{\kappa(W)\|\vec{z}\|}{\|W\|}
	\end{equation}
	and $\kappa(W)=\kappa^2$.
	
	Let us bound the terms in the parenthesis in the R.H.S. in (\ref{eq:deltaabound}).
	In order to bound the first term, we need another inequalities:
	\begin{equation}
	\|\vec{z}\| \le \sqrt{d}, \label{eq:ineq2}
	\end{equation}
	which follows from
	\begin{equation}
	0\le z_i = \frac{1}{N_D}\sum_{k=1}^{N_D}x^{(i)}_ky_k \le 1.
	\end{equation}
	Combining (\ref{eq:ineq1}), (\ref{eq:ineq3}) and (\ref{eq:ineq2}), we obtain
	\begin{equation}
	\frac{\|\vec{z}\|\cdot\|\delta W\|}{\|W\|^2} < \frac{d^{3/2}\epsilon^\prime}{c^2}. \label{eq:ineq4}
	\end{equation}
	
	Next, let us consider the second term $\|\delta\vec{z}\|/\|W\|$.
	This is bounded as
	\begin{equation}
	\frac{\|\delta\vec{z}\|}{\|W\|} < \frac{\sqrt{d}\epsilon^\prime}{c}.  \label{eq:ineq5}
	\end{equation}
	Here, we used
	\begin{equation}
	\|\delta \vec{z}\| \le \sqrt{d}\epsilon^\prime,
	\end{equation}
	which follows from the assumption $|\delta z_i|<\epsilon^\prime$, and (\ref{eq:ineq1}).
	
	Combining (\ref{eq:deltaabound}), (\ref{eq:ineq4}) and (\ref{eq:ineq5}), we obtain
	\begin{equation}
	\|\vec{\hat{a}}-\vec{a}\| \le \kappa^4\left(\frac{d^{3/2}\epsilon^\prime}{c^2}+\frac{d^{1/2}\epsilon^\prime}{c}\right) + O((\epsilon^{\prime})^2) \le \frac{2d^{3/2}\kappa^4}{c^2} \epsilon^\prime + O((\epsilon^{\prime})^2),
	\end{equation}
	where the second inequality holds since $d\ge 1$ and $c\le 1$.
	Noting that $\|\vec{\hat{a}}-\vec{a}\|_\infty \le \|\vec{\hat{a}}-\vec{a}\|$, we finally obtain (\ref{eq:erreval}) for $\epsilon^\prime$ satisfying $\epsilon^\prime < c^2\epsilon/2d^{3/2}\kappa^4$.
\end{proof}

\subsection{\label{sec:app2} Proof of Lemma \ref{lem:deltarel}}

\begin{proof}
	We can see that
	%\begin{widetext}
		\begin{eqnarray}
			\delta^\prime% \nonumber \\
			&= & \|\vec{a}^\prime-\vec{a}^\star\| \nonumber \\
			& = & \left\|\vec{a}^\star-\vec{a}+(\hat{H}_F(\vec{a}))^{-1}\vec{\hat{g}}_F(\vec{a})\right\| \nonumber \\
			& = & \left\|\vec{a}^\star-\vec{a}+(H_F(\vec{a}))^{-1}\vec{g}_F(\vec{a}) + \vec{\Delta}_F(\vec{a})\right\| \nonumber \\
			& = & \left\|(H_F(\vec{a}))^{-1} \cdot \left[H_F(\vec{a})\left(\vec{a}^\star-\vec{a}\right)+\vec{g}_F(\vec{a})-\vec{g}_F(\vec{a}^\star)\right] + \vec{\Delta}_F(\vec{a}) \right\| \nonumber \\
			& \le & \left\|(H_F(\vec{a}))^{-1}\right\| \cdot \left\|H_F(\vec{a})\left(\vec{a}^\star-\vec{a}\right)+\vec{g}_F(\vec{a})-\vec{g}_F(\vec{a}^\star)\right\| + \left\| \vec{\Delta}_F(\vec{a}) \right\| \nonumber \\
			& \le & \frac{M}{2\mu}\|\vec{a}-\vec{a}^\star\|^2+\epsilon \nonumber \\
			& = & \frac{M}{2\mu}\delta^2 + \epsilon.
		\end{eqnarray}
	%\end{widetext}
	Here, we used $\vec{g}_F(\vec{a}^\star)=\vec{0}$ to transform the third line to the fourth line and (\ref{eq:HInvEigenUB}),(\ref{eq:gDiffBound}) and (\ref{eq:ErrOrd}) to transform the fifth line to the sixth line.
	Using simple algebra, we can check that
	\begin{equation}
		\delta_-<\delta<\delta_+ \Rightarrow \frac{M}{2\mu}\delta^2+\epsilon < \delta,
	\end{equation}
	and
	\begin{equation}
		\delta\le\delta_- \Rightarrow \frac{M}{2\mu}\delta^2+\epsilon \le \delta_-,
	\end{equation}
	which means that the later part of the statement holds.
\end{proof}

\subsection{\label{sec:app3} Proof of Lemma \ref{lem:ErrNewConv}}

\begin{proof}
	
	First, let us see that
	\begin{equation}
		\delta_{n+1}\le 2\epsilon \ {\rm if} \ \delta_n\le 2\epsilon, \label{eq:lessthan2eps}
	\end{equation}
	for a non-negative integer $n$.
	This immediately follows Lemma \ref{lem:deltarel}.
	According to it, $\delta_{n+1}<\delta_n$ if $\delta_-<\delta_n\le2\epsilon$, and $\delta_{n+1}\le\delta_-$ if $\delta_n\le\delta_-$ (note that $\delta_-<2\epsilon<\delta_+$ holds under (\ref{eq:epsCond})).
	Combining these, we obtain (\ref{eq:lessthan2eps}).
	
	Then, let us consider the following three cases of the initial deviation.
	
	\ \\
	
	\noindent (i) $\delta_0\le 2\epsilon$
	
	From (\ref{eq:lessthan2eps}), we immediately see that $\delta_n \le 2\epsilon$ also for $n=1,2,...$.
	
	\ \\
	
	\noindent (ii) $2\epsilon<\delta_0 \le\sqrt{2\mu\epsilon/M}$ \ \footnote{Note that (\ref{eq:epsCond}) means $\sqrt{\frac{2\mu\epsilon}{M}}>2\epsilon$.}
	
	In this case, according to Lemma \ref{lem:deltarel}, $\delta_1\le\frac{M}{2\mu}(\delta_0)^2+\epsilon\le2\epsilon$.
	Therefore, by (\ref{eq:lessthan2eps}), $\delta_n\le2\epsilon$ holds for $n=1,2,...$.
	
	\ \\
	
	\noindent (iii) $\sqrt{2\mu\epsilon/M}< \delta_0<\frac{\mu}{M}$
	
	We define
	\begin{equation}
		\tilde{n}:=\left\lceil\log_2 \left(\frac{\log\left(\frac{2M\epsilon}{\mu}\right)}{2\log\left(\frac{M\delta_0}{\mu}\right)}\right)\right\rceil, \nonumber
	\end{equation}
	which is positive in this case.
	In the sequence of the updates from $\vec{a}_0$ to $\vec{a}_{\tilde{n}}$, if $\delta_{n^\prime}\le\sqrt{2\mu\epsilon/M}$ for some $n^\prime\in\{0,1,...,\tilde{n}-1\}$, $\delta_n\le2\epsilon$ holds for $n> n^\prime$, as we discussed in case (ii).
	Otherwise, $\delta_n>\sqrt{2\mu\epsilon/M}$, and therefore $\delta_{n+1}\le \frac{M}{2\mu}(\delta_n)^2+\epsilon\le\frac{M}{\mu}(\delta_n)^2$ holds %\footnote{Note that this represents the so-called quadratic convergence of Newton's method.}
	for $n=0,1,...,\tilde{n}-1$, which means that
	\begin{equation}
		\delta_n \le \frac{\mu}{M}\left(\frac{M\delta_0}{\mu}\right)^{2^n},
	\end{equation}
	for $n=1,...,\tilde{n}$ and especially $\delta_{\tilde{n}}\le\sqrt{2\mu\epsilon/M}$.
	Therefore, we see that $\delta_{n}\le2\epsilon$ for $n\ge\tilde{n}+1$, as we discussed in case (ii).
	
	In summary, the statement of the lemma has been proven in all cases.
	
\end{proof}

\subsection{\label{sec:app4} Proof of Lemma \ref{lem:comperr}}

\begin{proof}
	We let $\delta H_F(\vec{a})$ and $\delta\vec{g}_F(\vec{a})$ be the deviations of the estimated Hessian and gradient from the true values, that is, 
	\begin{equation}
		\hat{H}_F(\vec{a})=:H_F(\vec{a})+\delta H_F(\vec{a}), \vec{\hat{g}}_F(\vec{a})=:\vec{g}_F(\vec{a})+\delta\vec{g}_F(\vec{a}),
	\end{equation}
	respectively.
	Besides, we let $\delta H_{F,ij}(\vec{a})$ be the $(i,j)$-th component of $\delta H_F(\vec{a})$ and $\delta g_{F,i}(\vec{a})$ be the $i$-th component of $\delta\vec{g}_F(\vec{a})$.
	Furthermore, following the statement of the lemma, we assume that $\delta H_{F,ij}(\vec{a})$ and $\delta g_{F,i}(\vec{a})$ are bounded as follows
	\begin{equation}
		|\delta H_{F,ij}(\vec{a})|\le\epsilon^\prime_H, |\delta g_{F,i}(\vec{a})|\le\epsilon^\prime_g. \label{eq:EstimErr}
	\end{equation}
	
	Firstly, we show that, under (\ref{eq:epH}), or more loosely%\footnotemark[7]
	\begin{equation}
		\epsilon^\prime_H<\frac{\mu}{d}, \label{eq:condPosDef}
	\end{equation}
	$\hat{H}_F(\vec{a})$ is positive-definite and therefore the update (\ref{eq:NewtonErr}) from $\vec{a}$ is well-defined.
	Here, note that $\epsilon/4\tilde{\delta}\le\epsilon/4\delta_-<1/4$ since $\epsilon<\delta_-$, which can be shown by simple algebra under (\ref{eq:epsCond}).
	From (\ref{eq:condPosDef}), we obtain
	\begin{equation}
		\|\delta H_F(\vec{a})\|\le\|\delta H_F(\vec{a})\|_F < \mu \label{eq:delHNormUB}
	\end{equation}
	On the other hand, as (\ref{eq:WWhatEigenRel}), we can see that 
	\begin{equation}
		|\hat{\lambda}_d - \lambda_d| \le \|\delta H_F(\vec{a})\|, \label{eq:eigenErrUB}
	\end{equation}
	where $\lambda_d$ and $\hat{\lambda}_d$ are the minimum eigenvalues of $H_F(\vec{a})$ and $\hat{H}_F(\vec{a})$, respectively.
	Combining (\ref{eq:delHNormUB}), (\ref{eq:eigenErrUB}) and $\lambda_d>\mu$, we see
	\begin{equation}
		\hat{\lambda}_d \ge \lambda_d -  \|\delta H_F(\vec{a})\| > 0,
	\end{equation}
	which means that $\hat{H}_F(\vec{a})$ is positive-definite.
	
	When $\hat{H}_F(\vec{a})$ is invertible, considering up to the first order perturbation with respect to $\epsilon^\prime_H$ and $\epsilon^\prime_g$, we obtain
	\begin{eqnarray}
		& &  \vec{\Delta}_F(\vec{a}) \approx -(H_F(\vec{a}))^{-1}\delta H_F(\vec{a})(H_F(\vec{a}))^{-1}\vec{g}_F(\vec{a}) \nonumber \\
		& & \qquad\qquad\quad +(H_F(\vec{a}))^{-1}\delta\vec{g}_F(\vec{a}). \label{eq:1stPert}
	\end{eqnarray}
	Using the following inequalities
	\begin{equation}
		\|\delta H_F(\vec{a})\| \le \|\delta H_F(\vec{a})\|_F \le d\epsilon^\prime_H,
	\end{equation}
	\begin{equation}
		\|\delta\vec{g}_F(\vec{a})\|\le d^{1/2}\epsilon^\prime_g
	\end{equation}
	and (\ref{eq:HInvEigenUB}), we can evaluate (\ref{eq:1stPert}) as
	\begin{eqnarray}
		& &  \|\vec{\Delta}_F(\vec{a})\| \nonumber \\
		& \lesssim & \|(H_F(\vec{a}))^{-1}\| \cdot \|\delta H_F(\vec{a})\| \cdot \|(H_F(\vec{a}))^{-1}\vec{g}_F(\vec{a})\|  \nonumber \\
		& & \qquad + \|(H_F(\vec{a}))^{-1}\| \cdot \|\delta\vec{g}_F(\vec{a})\| \nonumber \\
		& \le & \frac{d\epsilon^\prime_H}{\mu}\|(H_F(\vec{a}))^{-1}\vec{g}_F(\vec{a})\| + \frac{d^{1/2}\epsilon^\prime_g}{\mu}. \label{eq:errEvalInter}
	\end{eqnarray}
	Since
	\begin{equation}
		\|(H_F(\vec{a}))^{-1}\vec{g}_F(\vec{a})\| \le 2\delta \label{eq:upDiffUB}
	\end{equation}
	as shown later, we obtain $\|\vec{\Delta}_F(\vec{a}_n)\| \le\epsilon$ from (\ref{eq:errEvalInter}) for $\epsilon^\prime_H$ and $\epsilon^\prime_g$ satisfying (\ref{eq:epH}) and (\ref{eq:epg}), respectively.
	
	The remaining task is to show (\ref{eq:upDiffUB}).
	Note that $(H_F(\vec{a}))^{-1}\vec{g}_F(\vec{a})$ is the update difference from $\vec{a}$ with the errorless update (\ref{eq:Newton}).
	In light of this, let us define $\vec{a}^\prime_{\rm NoErr}$ as the result of the update (\ref{eq:Newton}) from $\vec{a}$ and $\delta^\prime_{\rm NoErr}:=\|\vec{a}^\prime_{\rm NoErr}-\vec{a}^\star\|$.
	Since $\delta<2\mu/M$, $\delta^\prime_{\rm NoErr} \le \delta$ holds by Lemma \ref{lem:deltarel} (Eq. (\ref{eq:delpr})) with $\epsilon=0$.
	Therefore,
	\begin{eqnarray}
		\|(H_F(\vec{a}))^{-1}\vec{g}_F(\vec{a})\|&=&\|\vec{a}^\prime_{\rm NoErr}-\vec{a}\|\nonumber \\
		&\le& \|\vec{a}^\prime_{\rm NoErr}-\vec{a}^\star\| + \|\vec{a}-\vec{a}^\star\| \nonumber \\
		& = & \delta^\prime_{\rm NoErr} + \delta \nonumber \\
		&\le& 2\delta.
	\end{eqnarray}

\end{proof}

\subsection{\label{sec:app5} Proof of Lemma \ref{lem:comperr2}}

\begin{proof}
	Assume that $\delta_n=\|\vec{a}_n-\vec{a}^\star\|\le\tilde{\delta}_0$.
	Then, $\epsilon^\prime_H\le\mu\epsilon/4\delta_nd$ because of (\ref{eq:epH2}).
	This means that the conditions on estimation errors (\ref{eq:epg}) and (\ref{eq:epH}) in Lemma \ref{lem:comperr} are satisfied for the update from $\vec{a}_n$.
	Since $\delta_n<\delta_+\le2\mu/M$ is also satisfied, (\ref{eq:ErrOrd2}) holds by Lemma \ref{lem:comperr}.
	Then, we see that $\delta_{n+1}\le\max\{\delta_n,\delta_-\}$ is satisfied by Lemma \ref{lem:deltarel}, which means $\delta_{n+1}\le\tilde{\delta}_0$.
	Using induction, we see that (\ref{eq:ErrOrd2}) holds for any $n$.
\end{proof}

\end{document}